\documentclass[10pt, draftcls, journal, letterpaper, onecolumn]{IEEEtran}
\usepackage{bbm}
\usepackage[draft]{hyperref}
\usepackage{pgfplots}

\usepackage{mathtools}
\usepackage{latexsym}
\usepackage{pictex}
\usepackage{fancyvrb}
\usepackage{fancyhdr}
\usepackage{color}
\usepackage{xfrac}
\usepackage{subfig}
\usepackage{rawfonts}
\usepackage{graphics}
\usepackage{graphicx}
\usepackage{amssymb,amsmath,amsthm,amsfonts}
\usepackage{ifthen}
\usepackage{bbm}
\usepackage{tikz}
\usepackage{cite}
\usepackage{enumerate}
\usepackage{algorithm}
\usepackage[]{algpseudocode}
\usepackage{verbatim}
\usepackage{balance}

\newtheorem{theorem}{Theorem}
\newtheorem{cor}[theorem]{Corollary}
\newtheorem{lemma}[theorem]{Lemma}

\newtheorem{remark}[theorem]{Remark}
\newtheorem{definition}[theorem]{Definition}

\newtheorem{example}[theorem]{Example}

\makeatletter
\def\BState{\State\hskip-\ALG@thistlm}
\makeatother

\normalsize

\begin{document}


\title{Coded Caching in a Multi-Server System \\ with Random Topology}
\author{
\IEEEauthorblockN{Nitish Mital, Deniz G{\"u}nd{\"u}z and Cong Ling}\\
\vspace*{-1.65cm}
}

\maketitle


\begin{abstract}
Cache-aided content delivery is studied in a multi-server system with $P$ servers and $K$ users, each equipped with a local cache memory. In the delivery phase, each user connects randomly to any $\rho$ out of $P$ servers. Thanks to the availability of multiple servers, which model small-cell base stations (SBSs), demands can be satisfied with reduced storage capacity at each server and reduced delivery rate per server; however, this also leads to reduced multicasting opportunities compared to the single-server scenario. A joint storage and proactive caching scheme is proposed, which exploits coded storage across the servers, uncoded cache placement at the users, and coded delivery. The delivery \textit{latency} is studied for both \textit{successive} and \textit{parallel} transmissions from the servers. It is shown that, with successive transmissions the achievable average delivery latency is comparable to the one achieved in the single-server scenario, while the gap between the two depends on $\rho$, the available redundancy across the servers, and can be reduced by increasing the storage capacity at the SBSs. The optimality of the proposed scheme with uncoded cache placement and MDS-coded server storage is also proved for successive transmissions. 
\end{abstract}
\let\thefootnote\relax\footnotetext{Part of this work was presented at the IEEE Wireless Communications and Networking Conference (WCNC) in 2018 \cite{nmital}. \\
This work was supported in part by the European Union`s H2020 research and innovation programme under the Marie Sklodowska-Curie Action SCAVENGE (grant agreement no. 675891), and by the European Research Council (ERC) Starting Grant BEACON (grant agreement no. 677854).\\
The authors are with the Department of Electrical and Electronic Engineering, Imperial College London (Email: \{n.mital, d.gunduz, c.ling\}@imperial.ac.uk ).}

\section{Introduction}\label{sec:intro}

Coded caching and distributed storage have received significant attention in recent years to exploit the available memory space and processing power of individual network nodes to increase the throughput and efficiency of data availability. With proactive caching, part of the data can be pushed to nodes' local cache memories during off-peak hours, called the \textit{placement phase}, to reduce the burden on the network during peak traffic periods, called the \textit{delivery phase} \cite{maddah} - \cite{comb1}. A different type of coded caching also improves the delivery performance in the so-called ``femtocaching'' scenario \cite{femtocaching}, where multiple cache-equipped small-cell base stations (SBSs) collaboratively deliver contents to users. Coding for distributed storage systems has been extensively studied in the literature (see, for example, \cite{dimakis_networkcodes}), and in the femtocaching scenario, ideal maximum distance separable (MDS) codes allow users to recover contents by collecting parity bits from only a subset of SBSs they connect to \cite{femtocaching}.

\begin{figure}
\centering
\begin{tabular}{@{}c@{}}
       \quad \quad \quad \ \  \includegraphics[width=0.22\textwidth]{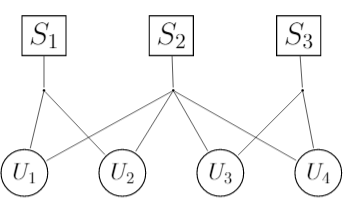}  \label{fig:f11}\\
        \small $(a)$ $\rho=2$, $q_1=2, q_2=4, q_3=2$.
    \end{tabular}
\begin{tabular}{@{}c@{}}
        \includegraphics[width=0.24\textwidth]{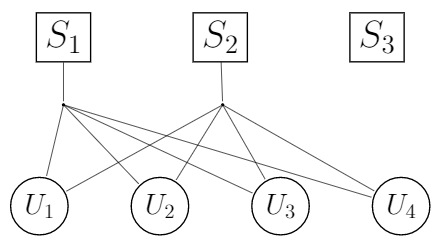} \label{fig:f12}\\
        \small $(b)$ $\rho=2$, $q_1=4, q_2=4, q_3=0$ (best topology (for successive transmissions), worst topology (for parallel transmissions)).
    \end{tabular}
\begin{tabular}{@{}c@{}}
        \includegraphics[width=0.22\textwidth]{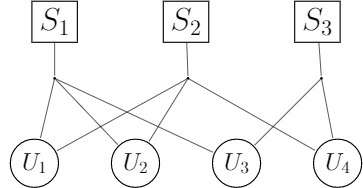}\label{fig:f13} \quad \quad 
        \includegraphics[width=0.22\textwidth]{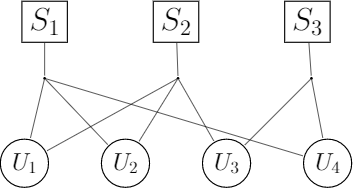}\label{fig:f1}\\
        \small $(c)$ $\rho=2, q_1=3, q_2=3, q_3=2$ (worst topology (for successive transmissions), best topology (for parallel transmissions))
    \end{tabular}
\caption{Examples of different network topologies for $P=3$ and $K=4$ with $\rho=2$.}
\label{fig:f1}
\end{figure}

In this work, we combine distributed storage at the SBSs, similar to the ``femtocaching'' framework \cite{femtocaching}, with cache storage at the users, and consider coded delivery over error-free shared broadcast links \cite{nmital}. We consider a library of $N$ files stored across $P$ SBSs, each equipped with a limited-capacity storage space (see Fig. \ref{fig:f1}). Unlike the existing literature, we consider a random connectivity model: during the delivery phase, each user connects only to a random subset of $\rho$ SBSs, where $\rho \leq P$. This may be due to physical variations in the channel, or due to resource constraints. Most importantly, the connections that form the network topology are not known in advance during the placement phase; therefore, the cache placement cannot be designed for a particular network topology. Storing the files across multiple SBSs, and allowing users to connect randomly to a subset of them results in a loss in multicasting opportunities for the servers, indicating a trade-off between the coded caching gain and the flexibility provided by distributed storage across the servers, which, to the best of our knowledge, has not been studied before. 

On the other hand, the presence of multiple servers may improve the latency if user requests can be satisfied in parallel. Accordingly, two scenarios are discussed depending on the delivery protocol. If the servers transmit \textit{successively}, i.e., time-division transmission, the total latency is the sum of the latencies on each link in delivering all the requests. If the servers operate in parallel, then the latency is given by the link with the maximum latency. 

We propose a practical coded storage and delivery scheme that exploits MDS coded storage across servers simultaneously with coded caching and delivery to users. In the successive transmission scenario, we show that the cost of the flexibility of distributed storage is a scaling of the latency by a constant. We also characterize the average worst-case latency (over all user-server associations) of the proposed scheme by assuming that the users connect to a uniformly random subset of the servers; and show that it is relatively close to the best-case performance, which is the single-server centralized delivery latency derived in\cite{maddah}, achieved when all the users connect to the same set of servers, maximizing the multicasting opportunities. We observe that, as the server storage capacities increase, the average delivery latency vs. user cache memory trade-off improves, approaching the single-server performance. We give an analytical expression to compute the average delivery latency for different server storage capacities, which is shown to give a fairly accurate estimate of the expected delivery latency when the number of servers is large. We then consider the delivery latency when the servers can transmit in parallel. We characterize the achievable average worst case delivery latency of the proposed coded storage and delivery scheme as a function of the server storage capacity for different $\rho$ values. 

In a related work \cite{Vaneet}, the authors study coded caching schemes presented in \cite{maddah} and \cite{chao} when parity servers are available. The authors consider special scenarios with one and two parity servers. They propose a scheme that stripes the files into blocks, and codes them across the servers with a systematic MDS code, and they also propose a scheme for the scenario in which files are stored as whole units in the servers, without striping. In our work, we do not specify servers as parity servers, and instead propose a scheme that generalizes to the use of any type of MDS code and any number of storage servers. We study the impact of the topology on the sum and maximum delivery rates, and the trade-off between the server storage capacity and the average of these rates. 

In \cite{shariatpanahi}, the authors consider multiple servers, each having access to all the files in the library, serving the users through an intermediate network of relays. They consider the so-called \textit{linear network} model, in which the network topology is fixed but unknown at the relay nodes. The authors study the delivery latency considering parallel transmissions from the servers, and show that there is a gain from using multiple servers when the relay nodes employ simple random linear network coding. Note that, our model considers both limited storage servers and random network topology over the delivery network, which is unknown during the placement phase, but known during the delivery phase. Compared to the linear network model, our model corresponds to an identity network transfer matrix, in which the scheme of \cite{shariatpanahi} does not provide any gains, since it is not optimized for the realization of the topology. 


Another line of related works study caching in combination networks \cite{comb1},\cite{comb3}, which consider a single server serving cache-equipped users through multiple relay nodes. The server is connected to these relays through unicast links, which in turn serve a distinct subset of a fixed number of users through unicast links. A combination network with cache-enabled relay nodes is considered in \cite{comb3}. In our paper, we relax the symmetry of a standard combination network and the assumption of a fixed and known network topology, which would be unrealistic in many practical scenarios, to a certain degree by allowing each user to connect to a random fixed number of servers, thus breaking the symmetry from the servers' perspective while maintaining the symmetry from the end-users' perspective.  \\

\textbf{\textit{Notations}}. For two integers $i<j$, we denote the set $\{i,i+1,\ldots, j \}$ by $[i:j]$, while the set $[1:j]$ is denoted by $[j]$. Sets are denoted with the calligraphic font, and $|\mathcal{A}|$ denotes the cardinality of set $\mathcal{A}$. For $\mathcal{A} =\{a_1,a_2,\ldots,a_p\}$, we define $X_{\mathcal{A}} \triangleq (X_{a_1}, \ldots, X_{a_p})$. $\mathbbm{1}_{E}$ denotes the indicator function of the event $E$, i.e., its value is $1$ when the event $E$ happens. $\lfloor x \rfloor$ denotes the largest integer less than or equal to $x$. $\lceil x \rceil$ denotes the smallest integer greater than or equal to $x$. 

\section{Problem Setting}\label{sec:app:1}
We consider the system model illustrated in Fig. \ref{fig:f1} with $P$ servers, denoted by $\mathrm{S}_1, \mathrm{S}_2, \ldots, \mathrm{S}_P$, serving $K$ users, denoted by $\mathrm{U}_1,\mathrm{U}_2, \ldots, \mathrm{U}_K$. There is a library of $N$ files $W_1, W_2, \ldots, W_N$, each of length $F$ bits uniformly distributed over $[2^{F}]$. Each user has access to a local cache memory of capacity $M_UF$ bits, $0 \leq M_U \leq N$, while each server has a storage memory of capacity $M_SF$ bits. The caching scheme consists of two phases: placement phase and delivery phase. We consider a centralized placement scenario as in \cite{maddah}, which is carried out centrally with the knowledge of the servers and the users participating in the delivery phase. However, neither the user demands, nor the network topology is known in advance during the placement phase. In the delivery phase, we assume that each user randomly connects to $\rho$ servers out of $P$ with a uniform distribution over all $\rho-$subsets, where $\rho \leq P$, and requests a single file from the library. We define $\alpha\triangleq \frac{\rho}{P}$ as the \textit{connectivity} of the network, where $0 \leq \alpha \leq 1$. For $j \in [K]$, let $\mathcal{Z}_j$ denote the set of servers $\mathrm{U}_j$ connects to, where $|\mathcal{Z}_j|=\rho$, and $d_j \in [N]$ denotes the index of the file it requests. For example, in Fig. \ref{fig:f1}(a), $\mathcal{Z}_1=\{ \mathrm{S}_1, \mathrm{S}_2 \}$, $\mathcal{Z}_2=\{ \mathrm{S}_1, \mathrm{S}_2 \}$, $\mathcal{Z}_3=\{ \mathrm{S}_2, \mathrm{S}_3 \}$ and $\mathcal{Z}_4=\{ \mathrm{S}_2, \mathrm{S}_3 \}$. Let the demand vector be denoted by $\mathbf{d}\triangleq(d_1, d_2,..., d_K)$. The topology of the network, i.e., which users are connected to which servers, and the demands of the users are revealed to the servers at the beginning of the delivery phase. 

\begin{figure}
\begin{center}
\includegraphics[trim=-2cm 10cm 0cm 1.8cm,clip,scale=0.40]{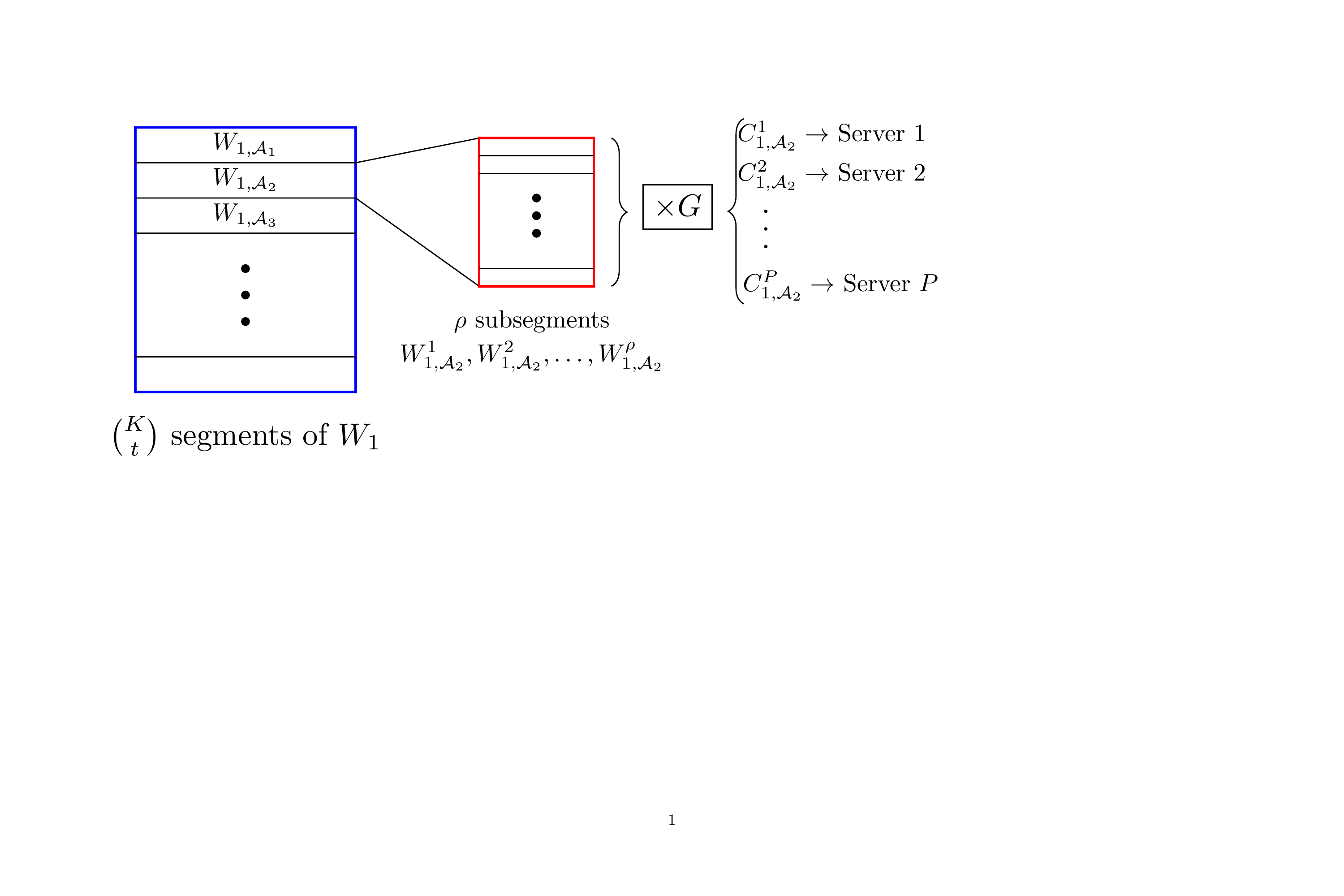}
\caption{Segmentation, MDS coding and placement of files.\label{fig:seg}}
\end{center}
\end{figure}

The complete library must be stored at the servers in a coded manner to provide redundancy, since each user connects only to a random subset of the servers. Since any user should be able to reconstruct any requested file from its own cache memory and the servers it is connected to, the total cache capacity of a user and any $\rho$ servers must be sufficient to recover the whole library; that is, we must have $M_U+\rho M_S\geq N$.

Let $\mathcal{K}_p$ denote the set of users served by $\mathrm{S}_p$, for $p\in [P]$, and define the random variable $Q_p \triangleq \left| \mathcal{K}_p \right|$, which denotes the number of users served by $\mathrm{S}_p$. We shall denote a particular realization of $Q_p$ as $q_p$ and define $\mathbf{q}\triangleq ( q_1, \ldots , q_p )$, where we have $\sum_{p=1}^{P} q_p = K\rho$. For example, in Fig. \ref{fig:f1}(a), we have $\mathcal{K}_1=\{ \mathrm{U}_1,\mathrm{U}_2 \}, \mathcal{K}_2=\{ \mathrm{U}_1, \mathrm{U}_2, \mathrm{U}_3, \mathrm{U}_4\}, \mathcal{K}_3=\{ \mathrm{U}_3, \mathrm{U}_4 \}$, and $\mathbf{q}=(2,4,2)$. In the delivery phase, server $\mathrm{S}_p$ transmits message $X_p$ of size $R_p F$ bits to the users connected to it, i.e., the users in set $\mathcal{K}_{p}$, over the corresponding shared link. We assume that each server is allocated a separate orthogonal delivery channel, and the message it transmits is received by all the users connected to this server. The message $X_{p}$ is a function of the demand vector $\mathbf{d}$, the network topology, the storage contents of server $\mathrm{S}_p$, and the cache contents of the users in $\mathcal{K}_p$. User $U_k$ receives the messages $X_{\mathcal{Z}_k} \triangleq \{ X_p : p \in \mathcal{Z}_k \}$, and reconstructs its requested file $W_{d_k}$ using these messages and its local cache contents.

\subsection{Formal Problem Statement}
We now provide the formal definition of the caching problem. Let $\{W_n\}_{n=1}^{N}$ be $N$ independent random variables each uniformly distributed over $[2^F]$ for some $F \in \mathbb{N}$. Each $W_n$ represents a file of size $F$ bits. Let $R_p, p\in [P],$ be the number of bits, normalized by the size of a file, transmitted by server $p\in [P]$ during the delivery phase. A $(M_S,M_U,R_1,\ldots,R_P)$ storage and caching scheme consists of $P$ server storage functions, $K$ caching functions, $P{P\choose \rho}^K$ encoding functions, and $K{P\choose \rho}^K$ decoding functions.

The caching function
\begin{align}
    \phi_k : [2^F]^N \rightarrow [2^{\lfloor FM_U\rfloor}],\quad  k \in [K],
\end{align}
maps the library $\{ W_n \}_{n=1}^{N}$ into the cache contents, of user $U_k$ during the placement phase:
\begin{align}
    V_k \triangleq \phi_k (W_1,\ldots, W_N)
\end{align}
The server storage function
\begin{align}
    \sigma_{p}: [2^F]^N \rightarrow [2^{\lfloor FM_S\rfloor}], \quad  p\in [P]
\end{align}
maps the library $\{W_n\}_{n=1}^{N}$ into the storage of server $S_p$:
\begin{align}
    Y_p \triangleq \sigma_p (W_1,\ldots,W_N).
\end{align}
We define a separate encoding function for each server depending on the network topology. Hence, the encoding function for server $S_p, p\in [P],$ 
\begin{align}
    \psi^{p}_{\{\mathcal{K}_p\}_{p=1}^{P}}: [N]^K \times [2^{\lfloor FM_S\rfloor}] \rightarrow [2^{\lfloor FR_p\rfloor}]
\end{align}
maps the demand vector and the memory contents of server $S_p$ to message $X_p, $ i.e.,
\begin{align}
    X_{p} \triangleq  \psi^{p}_{\{\mathcal{K}_p\}_{p=1}^{P}}(\mathbf{d},Y_p),
\end{align}
which is delivered to the users in $\mathcal{K}_p$ during the delivery phase. Finally, we define a separate decoding function for each user depending on the network topology. Hence, the decoding function for user $U_k, k\in [K],$ is
\begin{align}
    \mu^k_{\left(\{\mathcal{Z}_k\}_{k=1}^{K}\right)} :[N]^K \times [2^{\lfloor FR_{\pi^k(1)} \rfloor}] \times \cdots \times [2^{\lfloor FR_{\pi^k(\rho)} \rfloor}] \times [2^{\lfloor FM_U \rfloor}] \rightarrow [2^{\lfloor F \rfloor}],
\end{align}
where $\pi^k(1),\ldots, \pi^k(\rho)$ denote the $\rho$ servers in set $\mathcal{Z}_k$, maps the demand vector $\mathbf{d}$, the received signals $X_{\mathcal{Z}_k}$ from the servers in $\mathcal{Z}_k$, and the local cache content $V_k$ to the estimate $\hat{W}_{d_k},$ i.e.,
\begin{align}
    \hat{W}_{d_k} \triangleq \mu^k_{\left(\{\mathcal{Z}_k\}_{k=1}^{K}\right)} (\mathbf{d},X_{\mathcal{Z}_k},V_k)
\end{align} 
The probability of error for this scheme, for a fixed topology, is defined as 
\begin{align}
    \max_{\mathbf{d}\in [N]^K} \max_{k\in [K]} \mathrm{Pr}(\hat{W}_{d_k} \neq W_{d_k}).
\end{align}
We remark here that the storage and caching functions $\sigma_p$ and $\phi_k$ do not depend on the network topology, while the encoding and decoding functions do.
\begin{definition}
The tuple $(M_S,M_U,R_1,\ldots,R_P)$ is said to be achievable if for every $\epsilon >0$ and large enough file size $F$ there exists a $(M_S,M_U,R_1,\ldots,R_P)$ caching scheme with probability of error less than $\epsilon$.
\end{definition}

Our goal is to minimize the delivery latency, which is the time by which all the user requests can be satisfied. Among other parameters, delivery latency also depends on the operation of the SBSs. If each SBS transmits over an orthogonal frequency band, the requests can be delivered in parallel, and the delivery latency is given by $T_{pd}=\max_{p} R_p$. If, instead, the servers transmit successively in a time-division manner, which is suitable for user devices that are simple and not capable of multihoming on multiple frequencies, the normalized delivery latency will be given by $T_{sd}=\sum_{p=1}^{P} R_p$. Our goal will be to find the average worst-case delivery latency, where the worst case refers to the fact that all the users can correctly decode their requested files, independent of the combination of files requested by them, and the averaging is over all possible network topologies. Assuming that $N\geq K$ (i.e., the number of files is larger than the number of users), it is not difficult to see that all the users requesting a different file corresponds to the worst-case scenario. We would also like to remark that, under uniform file popularity, the probability of experiencing this worst-case demand distribution increases significantly with $N$, and approaches $1$ for $N$ values that one expects to experience in practice. 

\section{Coded Distributed Storage and Caching Scheme}
We first note that our system model brings together aspects of distributed storage and proactive caching/coded delivery. To see this, consider the system without any user caches, i.e., $M_U=0$, which is equivalent to a distributed storage system with unreliable servers, where random $P-\rho$ out of $P$ servers are inactive. It is known that MDS codes provide much higher reliability and efficiency compared to replication in this scenario \cite{dimakis_networkcodes}. On the other hand, when the servers are reliable, i.e., $\rho = P$, our system is equivalent to the one in \cite{maddah}, and coded delivery provides significant reductions in the delivery latency. Accordingly, our proposed scheme brings together benefits from coded storage and coded delivery. To illustrate the main ingredients of the proposed scheme we assume $M_S=\frac{N}{\rho}$ in this section, and extend to other server capacities in later sections.

\subsection{Server Storage Placement } \label{sec:serverstorage}
We first describe how the files are stored across the SBS servers in order to guarantee that each user request can be satisfied from any $\rho$ servers a user may connect to (see Fig. \ref{fig:seg}). We define $t \triangleq \frac{KM_U}{N}$, and assume initially that $t$ is an integer, i.e., $t\in [0:M_U]$. The solution for non-integer $t$ values will be obtained through memory-sharing \cite{maddah}. Each file is divided into $K\choose t$ equal-size non-overlapping segments. We enumerate them according to distinct $t$-element subsets of $[K]$, where $W_{j, \mathcal{A}}$ denotes the segment of $W_j$ that corresponds to subset $\mathcal{A}$. We have $W_j = \bigcup_{\mathcal{A} \subset [K]: |\mathcal{A}| = t} W_{j, \mathcal{A}}$, $j \in [N]$.

Each segment is further divided into $\rho$ equal-size non-overlapping sub-segments denoted by $W_{j,\mathcal{A}}^l$, $l\in[\rho]$. The $\rho$ sub-segments of each segment are coded together using a $(P,\rho)$ linear MDS code with generator matrix $G$, giving as output $P$ coded subsegments for segment $W_{j,\mathcal{A}}$, denoted by $C_{j,\mathcal{A}}^{l}, l\in [P]$. $C_{j,\mathcal{A}}^{l}$ is a linear combination of the subsegments of the segment corresponding to subset $\mathcal{A}$, of file $W_j$. $C_{j,\mathcal{A}}^{l}$ will be stored in server $\mathrm{S}_l$, $\forall l \in [P],j\in [N]$, and $\mathcal{A}\subset [K], \vert \mathcal{A} \vert=t$.
Since each sub-segment is of length $\frac{F}{\rho{K\choose t}}$, every linear combination $C_{j,\mathcal{A}}^{l}$ is of the same length; and hence, server storage capacity constraint of $M_S F=\frac{NF}{\rho}$ is met with equality.
\begin{remark}We assume that each user knows the generator matrix of the MDS code to be able to reconstruct any coded subsegment $C_{j,\mathcal{A}}^{l}$ from the uncoded segment $W_{j,\mathcal{A}}$. \end{remark}
\subsection{User Cache Placement}
Using the placement scheme proposed in \cite{maddah} for user caches, each segment of a file, $W_{j,\mathcal{A}}$, is placed into the caches of all the users $\mathrm{U}_k$ for which  $k\in \mathcal{A}$, i.e., each user caches ${K-1 \choose t-1}$ segments of each file, or $\frac{{K-1 \choose t-1}}{{K \choose t}}NF = \frac{t}{K}N = M_U F$ bits, meeting the user cache capacity constraint.
\subsection{Delivery Phase } \label{sec:delivery}
We first make the following observation about the above placement scheme: in the worst-case demand scenario, consider any $t+1$ users. Any $t$ out of these $t+1$ users share in their caches one segment of the file requested by the remaining user. Enumerate these subsets of $t+1$ users as $\mathcal{H}_i$, $i\in \left[{K\choose t+1}\right]$.  Consider server $\mathrm{S}_p$, $p \in [P]$, and one of the $q_p$ users connected to it, say $U_k$. Then, for any subset $\mathcal{H}_i$, that includes $k$, i.e., $k \in \mathcal{H}_i$, the segment $W_{d_k,\mathcal{H}_i \setminus \{k\}}$ is needed by user $\mathrm{U}_k$, but is not available in its cache because $k \notin  \mathcal{H}_i \setminus \{k\}$, while it is available in the caches of the users in $\mathcal{K}_p \bigcap \mathcal{H}_i \setminus \{k \}$. The MDS coded subsegment of $W_{d_k, \mathcal{H}_i \backslash \{k\}} $ stored by $\mathrm{S}_p$ is $C_{d_k, \mathcal{H}_i \backslash \{k\}}^{p}$, and since the users know the generator matrix $G$, each user which has $W_{d_k, \mathcal{H}_i \backslash \{k\}} $ in its cache can reconstruct $C_{d_k, \mathcal{H}_i \backslash \{k\}}^{p} $ as well. Then, for each $\mathcal{H}_i$ that includes at least one user from $\mathcal{K}_p$, $\mathrm{S}_p$ transmits
\begin{align} \label{Tx}
    X_p(\mathcal{H}_i) = \bigoplus_{k \in \mathcal{K}_p \bigcap \mathcal{H}_i \setminus \{k\}} C_{d_k,\mathcal{H}_i \backslash \{k\}}^{p} ,
\end{align}
where $\bigoplus$ denotes the bitwise XOR operation. Then, $\Bigl| \left\{i \in \left[{K\choose t+1}\right]: k \in \mathcal{H}_i \right\} \Bigr| = {K-1 \choose t} $ is the number of messages transmitted by server $\mathrm{S}_p$ that contain the coded version of a segment requested by $\mathrm{U}_k$, and is also equal to the number of segments of $W_{d_k}$ not present in the cache of user $\mathrm{U}_k$. Overall, the message transmitted by $S_p$ is given by
\begin{align} \label{Tx_all}
   X_p = \bigcup_{i \in \left[{K\choose t+1}\right]: \mathcal{K}_p \bigcap \mathcal{H}_i\neq \phi} X_p( \mathcal{H}_i).
\end{align}
From the transmitted message $X_p(\mathcal{H}_i)$ in \eqref{Tx} for each set $\mathcal{H}_i$, user $\mathrm{U}_k$ can decode the MDS coded version $C_{d_k,\mathcal{H}_i \backslash \{ k\}}^{p}$ of its requested segment $W_{d_k, \mathcal{H}_i \backslash \{k\}}$. With the transmissions from all the servers, $\mathrm{U}_k$ receives $\rho$ coded versions of each missing segment from the $\rho$ servers it is connected to. Since each segment is coded with a $(P,\rho)$ MDS code, the user is able to decode each missing segment of its request.

Note that each transmitted message $X_p(\mathcal{H}_i)$ by a server is of length $\left. F \right/ \rho {K\choose t}$ bits. The number of messages transmitted by $\mathrm{S}_p$ is
\begin{align}
    \left| \left\{ i\in \left[{K\choose t+1}\right]: \mathcal{K}_p \bigcap \mathcal{H}_i \neq \phi \right\} \right|&= {K\choose t+1}-\left| \left\{ i\in \left[{K\choose t+1}\right]: \mathcal{K}_p \bigcap \mathcal{H}_i = \phi \right\} \right|\\
    &={K\choose t+1}-{K-q_p\choose t+1}.
\end{align} 
That is, server $\mathrm{S}_p$ transmits a total of $R_p=\sfrac{F}  {\rho {K\choose t}}\left[{K\choose t+1}-{K-q_p\choose t+1}\right]$ bits.

The delivery latency performance of this proposed coded storage and delivery scheme with both successive and parallel SBS transmissions will be studied in the following two sections.  
\begin{remark}
Due to the symmetry in the network across servers and users, the delivery latency of this scheme depends only on the $\mathbf{q}$ vector, not the particular network topology, i.e., what matters is the number of users served by each server, not the identity of the users. More specifically, all permutations of a $\mathbf{q}$ vector, and the associated users, result in the same latency. Hence, we define the ``type" of a network topology as a vector of dimension $K+1$, $\mathbf{g}$, where $g_i$ denotes the number of servers serving $i$ users, for $i=0,1,\ldots, K$. We have $0\leq g_i \leq P$, $\sum_{i=0}^{K} g_i = P$ and $\sum_{i=0}^{K} ig_i=K\rho$.
\end{remark}
\section{Successive SBS Transmissions}\label{s:Successive}
In this section we assume that the SBSs share the same communication resources, and hence, transmit successively to avoid interference. When the SBSs transmit successively in time, the normalized delivery latency is given by
\begin{align} 
   T_{sd} \triangleq \sum_{p=1}^P R_p &= \frac{1}{\rho {K\choose t}} \sum_{p=1}^{P} \left[{K\choose t+1}-{K-q_p\choose t+1} \right] \\
	&=\frac{1}{\alpha} \frac{(K-t)}{(t+1)} - \frac{1}{\rho {K\choose t}} \sum_{p=1}^{P}  {K-q_p \choose t+1} \label{eq:sumrate1} \\
	&= \frac{1}{\alpha} \frac{(K-t)}{(t+1)} - \frac{1}{\rho {K\choose t}} \sum_{i=0}^{K} g_i {K-i \choose t+1}.\label{eq:sumrate}
\end{align}

To characterize the ``best'' and ``worst'' network topologies that lead to the minimum and maximum delivery latency, respectively, we present the following lemma without proof.
\begin{lemma}\label{convexity}
For $n_1, n_2 , r \in \mathcal{Z}^{+}$ satisfying $r \leq n_1$ and $n_1 + 2 \leq n_2$, we have 
\begin{align}
    {n_1 \choose r} + {n_2 \choose r} \geq {n_1+1 \choose r} + {n_2-1 \choose r}.
\end{align}
\end{lemma}
The lemma above indicates the ``convex'' nature of the binomial coefficients in \eqref{eq:sumrate1}; that is, the points $(r, {r \choose r})$, $(r+1, {r+1 \choose r})$, \ldots, $(n_1+n_2-r, {n_1+n_2-r \choose r})$ form a convex region. From Lemma \ref{convexity}, it can be deduced that the second summation term in \eqref{eq:sumrate1} takes its minimum when $\max_{p}(q_p) \leq \min_{p}(q_p) + 1$, $p\in [P]$, i.e., the values of $q_p$ are as close to each other as possible. This corresponds to the class of topologies with the highest delivery latency (see Fig. \ref{fig:f1}(c) for an example). The topology that requires the minimum delivery latency of $T_{sd} =\frac{K-t}{t+1}$ is when $q_p$ is either $0$ or $K$ for each server, or equivalently, when all the users are connected to the same $\rho$ servers (see Fig. \ref{fig:f1}(b) for an example).

Next we study the average worst-case normalized delivery latency, where the average is taken over all possible network topologies. As we have seen above, the delivery latency depends on the topology, and for a given topology, the ``worst-case'' delivery latency refers to the worst-case demand combination when each user requests a different file. Note that, in the worst case, due to the symmetry in the network and the proposed caching and delivery scheme, the latency depends only on the type of the network topology. We further assume that the probability of having any network of the same type is the same. 

\begin{lemma}\label{lemma_prob}
Let $w_i$ be the probability of exactly $i$ users being served by a server; that is, $w_i=Pr\{ q_p=i \}, p\in [P]$. We have
\begin{align}
\mathbb{E}[g_i] = w_i P.
\end{align}
\end{lemma}
\begin{proof}
The number of servers serving exactly $i$ users, $g_i$, can be written as
\begin{align}
g_i=\sum_{p=1}^{P} \mathbbm{1}_{\{q_p=i\}}.
\end{align}
Taking expectation on both sides, we have
\begin{align}
\mathbb{E}[g_i] &= \sum_{p=1}^{P} \mathrm{Pr}\{ q_p=i \}\\
&=w_i P.
\end{align}
\end{proof}
The following theorem presents the average normalized worst-case delivery latency of the proposed scheme under successive transmissions, which follows by taking the expectation of both sides of Eq. \eqref{eq:sumrate} and Lemma \ref{lemma_prob}.
\begin{theorem}
The average worst-case normalized delivery latency of the proposed scheme over all topologies under random user-server association is given by 
\begin{align}
\mathbb{E}[T_{sd}]=\frac{1}{\alpha} \frac{(K-t)}{(t+1)} - \frac{1}{\alpha {K\choose t}}\sum_{i=0}^{K} w_i {K-i\choose t+1} .\label{eq:avgrate}
\end{align}
\end{theorem}

Since we have assumed uniform random connectivity, we have $w_i=\sfrac{{K\choose i}{P-1 \choose \rho -1}^{i} {P-1 \choose \rho}^{K-i}}{{P\choose \rho}^K}={K\choose i} \alpha^{i} (1-\alpha)^{K-i}$. The average worst-case latency is given in the following corollary.
\begin{cor}
The average worst-case normalized delivery latency with successive transmissions under uniformly random user-server association is given by 
\begin{align}
\mathbb{E}[T_{sd}]= \frac{K-t}{t+1} \left[ \frac{1-(1-\alpha)^{t+1}}{\alpha} \right].
\end{align}
\begin{proof}
By plugging in $w_i={K\choose i} \alpha^{i} (1-\alpha)^{K-i}$ in Eq. \eqref{eq:avgrate}, we obtain the above simplified expression.
\end{proof}
\end{cor}

\subsection{Redundancy in Server Storage Capacity} \label{sec:incservercache}

In the analysis above, we have set the server storage capacity to $M_S=\frac{N}{\rho}$.  On the other hand, for a given user cache capacity $M_U$, the minimum server storage capacity that would allow the reconstruction of any demand combination is given by $M_S=\frac{N-M_U}{\rho}$. In this case, we cache the same $\frac{M_U}{N}$ fraction of the library in all the user caches during the placement phase, and deliver the remaining fraction of the demands from the servers, which is identical to the scheme in \cite{comb3} when the user and its connected servers have just enough space to store the entire library. The worst-case delivery latency in this case is given by $T_{sd}=K\left( 1-\frac{M_U}{N}\right)= K-t$.\\
\indent Next, we consider the case when there is redundancy in server memories; that is, $\frac{N}{\rho}< M_S \leq N$. Assume that $M_S=\frac{N}{\rho - z}$ for some integer $ z\in [\rho -1]$. Define $ \hat{\alpha}\triangleq \frac{\rho-z}{P}  $. Since $\alpha$ is defined as the connectivity of the network, $\alpha-\hat{\alpha}$ is the storage redundancy. For non-integer values of $z$, the solution can be obtained by memory-sharing.

\begin{figure}
\begin{center}
\includegraphics[scale=0.3]{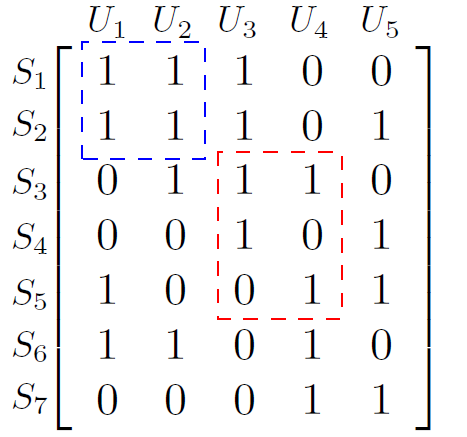}
\caption{An example $7 \times 5$ incidence matrix ($P=7, K=5$) with $\rho = 4$.}
\label{fig:conn_matrix}
\end{center}
\end{figure}

In this case, a $(P,\rho-z)$ MDS code is used for server storage placement, allowing each user to reconstruct any requested file by connecting to $\rho -z$ servers. The user cache placement is done as in the previous section. In the delivery phase, each user randomly connects to $\rho$ servers. We now have a degree of freedom thanks to the additional storage space available at each server. Each user can obtain a segment from any $\rho -z$ of the $\rho$ servers it is connected to by receiving one coded subsegment from each of them. The choice of the servers that deliver the coded subsegments to the users is made such that the multicasting opportunities across the network are maximized. We construct an incidence matrix $A$ of dimensions $P \times K$ such that $a_{ij}=1$ if $S_i$ is connected to $U_j$, $a_{ij}=0$ otherwise. Consider the $(t+1)-$element subset $\mathcal{H}_i$, and the file segments $W_{d_k,\mathcal{H}_i\backslash \{ k\}}, \forall k\in \mathcal{H}_i$. Consider the columns of $A$ corresponding to the users in $\mathcal{H}_i$ and the matrix $Q$ formed by them. Define the minimum cover of $\mathcal{H}_i$ as the smallest $l$ for which a $l\times (t+1)$ submatrix of $Q$ has at least $\rho-z$ non-zero values in each column. The servers corresponding to the $l$ rows of this submatrix have to transmit one coded message each to satisfy the requests for the missing segments corresponding to $\mathcal{H}_i$.  Therefore, the total number of transmissions required to deliver the segments $W_{d_k,\mathcal{H}_i\backslash \{ k\}}, k\in \mathcal{H}_i$, is equal to the minimum cover of $\mathcal{H}_i$. \\
\indent As an example, consider the incidence matrix as shown in Fig. \ref{fig:conn_matrix}, which corresponds to a system with $P=7$ servers and $K=5$ users, where each user connects to $\rho=4$ servers. Assume that the server storage capacity is $M_S=\frac{N}{\rho-2}$ and $t=1$. In this setting, coded subsegments of requested files can be delivered to $t+1=2$ users through multicasting, and it is sufficient for each user to receive coded segments from $\rho-2=2$ servers. Then, for the user set $\mathcal{H}_i=\{1,2\}$, we consider the submatrix corresponding to the columns $1$ and $2$ and rows $1$ and $2$ (marked by the blue dashed lines in Fig. \ref{fig:conn_matrix}), which is the smallest submatrix satisfying the condition that each column has at least $\rho-z=2$ $1$s. Hence, the minimum cover for $\mathcal{H}_i=\{ 1,2 \}$ is equal to the number of rows of this submatrix, that is, $2$.  For $\mathcal{H}_i=\{3,4\}$ (marked by the red dashed lines in Fig. \ref{fig:conn_matrix}), the minimum cover is $3$. Thus, from \eqref{Tx}, for segments $W_{d_k,\{3,4\} \backslash \{ k\}}, k\in \{3,4\}$, $S_3$ transmits the message $X_3(\{3,4\}) = \bigoplus_{k \in \{3,4\}} C_{d_k,\{3,4\} \backslash \{k\}}^{3}$, $S_4$ transmits $X_4(\{3,4\})=C_{d_3,\{4\}}^{4}$, and $S_5$ transmits $X_5(\{3,4\})=C_{d_4,\{3\}}^{5}$. The total number of transmissions is $3$. We can go through all the $(t+1)-$ element subsets of the users and identify for each of them the minimum cover. We note that in the successive transmission scenario, the total latency does not depend on the server transmitting each subsegment, since the contribution to the total latency is the same. In the above example servers $S_1$ and $S_6$ could also deliver the two coded subsegments to users $U_1$ and $U_2$. The selection of the servers matters in the case of parallel transmissions. 
\subsection{Performance analysis}\label{performance_analysis}
In this section, we derive an analytical expression for the expected delivery latency in the asymptotic regime, i.e., when $P \rightarrow \infty$, while $\alpha$ and $\hat{\alpha}$ are fixed. Consider a particular subset $\mathcal{H}$ of $t+1$ users. Define $\beta_i$ as the fraction of servers serving $i$ users in $\mathcal{H}$, $i=0,1,\ldots,t+1$. Thus, we have
\begin{align} \label{beta_def}
\beta_i = \frac{1}{P}\sum_{p=1}^{P} \mathbbm{1}_{\{ \vert \mathcal{H} \cap \mathcal{K}_p\vert = i\}}.
\end{align}
Taking expectation on both sides of Eq. \eqref{beta_def}, we have 
\begin{align}
\mathbb{E}[\beta_i]&=\frac{1}{P} \sum_{p=1}^{P} \mathbb{E}[\mathbbm{1}_{\{ \vert \mathcal{H} \cap \mathcal{K}_p\vert = i\}}] \\
&= \frac{1}{P}  \sum_{p=1}^{P} \mathrm{Pr}( \vert \mathcal{H} \cap \mathcal{K}_p\vert = i)\\
&= \mathrm{Pr}( \vert \mathcal{H} \cap \mathcal{K}_p\vert = i) \label{exp_beta_1}
\end{align}
\begin{align}
= {t+1 \choose i} \alpha^i (1-\alpha)^{(t+1-i)}, \label{exp_beta}
\end{align}
where \eqref{exp_beta_1} follows due to the symmetry across all the servers. By the law of large numbers, $\beta_i \rightarrow \mathbb{E}[\beta_i]$ for all $i\in [K]$, as $P \rightarrow \infty$. Also, the topology becomes symmetric across all users as $P \rightarrow \infty$, i.e., almost all user subsets of the same size are served by the same number of servers.
\begin{figure}
    \centering
    \includegraphics[scale=0.60]{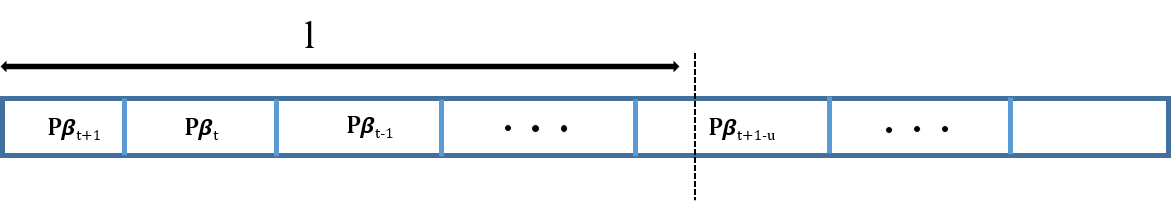}
    \caption{The ordering of servers to count the minimum cover. The dashed line indicates the point at which enough servers have been counted to deliver $\hat{\alpha}$ coded subsegments to all users in $\mathcal{H}$.}
    \label{fig:server_ord}
\end{figure}
We group the servers serving the same number of users and arrange them in the order as illustrated in Fig. \ref{fig:server_ord}, where the first $P\beta_{t+1}$ servers serve $t+1$ users in $\mathcal{H}$, the next $P\beta_{t}$ servers serve exactly $t$ users in $\mathcal{H}$, and so on. To compute the minimum cover $l$, i.e., the minimum number of servers that are needed to deliver $\hat{\alpha}$ coded subsegments to each user in $\mathcal{H}$, we start counting from the left, until each user in $\mathcal{H}$ collects $\hat{\alpha}$ coded subsegments. For some $u\in [0:t]$, we count till the $(u+1)$-th set of servers which serve $t+1-u$ users in $\mathcal{H}$, as in Fig. \ref{fig:server_ord}. When counting the set of servers serving $t+1-u$ users, note that, according to our scheme, the $t+1-u$ users can each extract one coded subsegment from a message transmitted by a server in that set. Therefore, $\lceil \frac{t+1}{t+1-u} \rceil$ servers are required to serve one coded subsegment each to the $t+1$ users in $\mathcal{H}$. Define $\delta$ as the number of coded subsegments required by a single user in $\mathcal{H}$ from the set of servers serving $t+1-u$ users. Therefore, for $P\rightarrow \infty$, the minimum cover can be written as
\begin{align} \label{min_cover}
l\approx P \sum_{j=0}^{u-1} \beta_{t+1-j} + \delta \left\lceil \frac{t+1}{t+1-u} \right\rceil 
\end{align}
for some $u\in [0:t]$, where \eqref{min_cover} follows thanks to the symmetry across users. Note that the above analysis is asymptotic, and does not hold in general for a finite $P$. 
Since a message transmitted by a server serving $i$ users in $\mathcal{H}$ delivers $i$ coded subsegments in total to the $i$ users, the total number of coded subsegments delivered by the $l$ servers that form the minimum cover for the users in $\mathcal{H}$ must be at least $(t+1)\hat{\alpha}$; that is, 
\begin{align} 
P\sum_{j=0}^{u-1}(t+1-j)\beta_{t+1-j} + \delta^{'} \left\lceil \frac{t+1}{t+1-u} \right\rceil \geq (t+1)P\hat{\alpha}, \label{total_segments}
\end{align}
where $\delta^{'}\triangleq (t+1-u)\delta$.
The value of $u$ is determined by solving for
\begin{align}
0 \leq (t+1)P\hat{\alpha}-P\sum_{j=0}^{u-1}(t+1-j)\beta_{t+1-j} \leq  (t+1-u)\beta_{t+1-u} .\label{determine_u}
\end{align}
From Eq. \eqref{exp_beta} and the asymptotic convergence of $\beta_i$ to its expectation, we have
\begin{align}
    \sum_{j=0}^{u-1} (t+1-j) \beta_{t+1-j}\overset{P\rightarrow \infty}{\rightarrow}& \sum_{j=0}^{u-1} {t+1\choose t+1-j} (t+1-j)\alpha^{(t+1-j)} (1-\alpha)^{j}\\
    =& \alpha^{t+1}(t+1)\sum_{j=0}^{u-1} {t\choose s}  \left( \frac{\alpha}{1-\alpha}  \right)^{-j} \label{solving_u}
\end{align}
We substitute \eqref{solving_u} into \eqref{determine_u} to solve for $u$. Having first determined $u$ from Eq. \eqref{determine_u}, and then $\delta$ from \eqref{total_segments}, we can find the minimum cover $l$ from Eq. \eqref{min_cover} for $P\rightarrow \infty$. The delivery latency can thus be estimated as 
\begin{align}\label{closedform_redundant}
    \mathbb{E}[T_{sd}]= \frac{1}{(\rho-z)}\left( \frac{K-t}{t+1}\right)l,
\end{align}
where the factor $\frac{1}{(\rho-z)}\left( \frac{K-t}{t+1}\right)$ is obtained by multiplying the normalized size of each coded subsegment, given by $\frac{1}{(\rho-z){K\choose t}}$, with the number of $(t+1)-$user subsets, given by ${K\choose t+1}$. It will be seen in Section \ref{s:Discussion} that Eq. \eqref{closedform_redundant} provides a fairly accurate estimate of the expected delivery latency when the number of servers $P$ is large. 

\section{Lower Bound}
In this section, we derive a tight lower bound on the minimum expected delivery latency with uncoded cache placement, coded distributed storage in the servers, and successive transmissions, which shows the optimality of the caching and delivery scheme proposed in Section \ref{s:Successive} in certain regimes. Following \cite{optimality_1}, we will first represent the problem as a set of index coding problems.\\
\indent In the index coding problem \cite{index_coding}, a sender wishes to communicate an independent message $M_j, j \in [B]$, uniformly distributed over $[2^{nr_j}]$, to the $j^{th}$ user among $B$ users by broadcasting a message $X^n$ of length $n$. Each user $j$ knows a subset of the messages targeting these $B$ users, indicated by $\mathcal{B}_j, \mathcal{B}_j \subset \{M_1,\ldots,M_B \}$, referred to as side information. A rate tuple $(r_1, \ldots,r_B)$ is achievable, for large enough $n$, if every user can restore its desired message with high probability based on $X^n$ and its side information. The index coding problem can be represented as a directed graph $G$ with $B$ nodes, where node $i$ represents message $M_i$, and a directed edge connects node $i$ to node $j$ if user $j$ knows message $M_i$ as side information. 
For our problem setting, where we have the file library $\{W_i \}_{i=1}^N$, each file $W_i, i\in [N]$, of size $F$ bits is divided into $2^{K}$ non-overlapping segments denoted by $W_{i,\mathcal{A}}$, $\mathcal{A}\in 2^{[K]}$, where $2^{[K]}$ indicates the power set $\{ \phi , \{1 \}, \{2 \}, \{ 3 \}, \{ 1,2 \}$ $, \ldots ,[K] \}$. The segment $W_{i,\mathcal{A}}$ denotes the part of file $W_i$ cached exclusively by users in set $\mathcal{A}$. This is the most general representation of an uncoded caching scheme at the users. For each demand vector $\mathbf{d}$ with distinct requests, corresponding to the worst case scenario, we consider an index coding problem with $K2^{K-1}$ independent messages, each of which represents a segment requested by a particular user and cached by a different subset of the remaining users.

We generate a directed graph with $K2^{K-1}$ nodes corresponding to these messages, such that, for $i \neq j$ and $\mathcal{A}_i \subset [K]\setminus \{ i \}$ and $\mathcal{A}_j \subset [K]\setminus \{ j \}$, there is a directed edge from node $W_{d_i,\mathcal{A}_i}$ to $W_{d_j,\mathcal{A}_j}$, $i\neq j$, if user $U_j$ caches the segment $W_{d_i, \mathcal{A}_i}$; that is, if $j \in \mathcal{A}_i$. 
In the single server centralized setting, we get a lower bound using the index coding bound \cite{optimality_2}. Multi-server index coding has been studied as the distributed index coding problem in \cite{distr_index_coding},\cite{cooperative_index_coding}. In the distributed index coding problem, the servers are considered to store a subset of the messages in uncoded form, and each user is connected to all the servers, whereas in our problem each user can connect to $\rho$ out of $P$ servers randomly. Therefore, for the user to be able to retrieve any requested file from the servers it connects to, the files must be stored using a distributed storage scheme in the servers. Therefore, we analyse the case where the files are stored using erasure codes in the servers. In that, we encode the segment $W_{d_i,\mathcal{A}}$ into $P$ distinct coded subsegments, denoted by $(C_{d_i,\mathcal{A}}^1,\ldots,C_{d_i,\mathcal{A}}^P) \in \prod_{p\in [P]}\left[ 2^{n_pr_{d_i,\mathcal{A}}^p} \right]$, where $n_p=FR_p$ is the length of message in bits transmitted by server $S_p$, such that any $\rho$ coded subsegments can be used to reconstruct the original segment. $r_{d_i,\mathcal{A}}^p$ is the rate at which server $S_p$ transmits the coded subsegment $C_{d_i,\mathcal{A}}^p$ corresponding to user $U_i's$ request, and we have $\sum_{j=1}^{\rho} n_{\pi(j)} r_{d_i,\mathcal{A}}^{\pi(j)} \geq \vert W_{d_i,\mathcal{A}} \vert$ as a necessary condition to ensure that the segment $\vert W_{d_i,\mathcal{A}} \vert$ can be reconstructed by receiving any $\rho$ distinct coded subsegments, where $\pi(j),j\in [\rho]$, are the $\rho$ servers in set $\mathcal{Z}_i$. Recall that $\mathcal{Z}_i$ is the set of $\rho$ servers that serve user $U_i$. Also note that we do not code across files, but encode each file separately.

For the multi-server scenario, we consider $P$ index coding problems, each represented as a distinct directed graph $G_p, p\in [P]$. Each node in $G_p$ corresponds to a distinct coded subsegment $C_{d_i,\mathcal{A}}^{p}$, which is requested by user $U_i$ and available in server $S_p$. By distinct coded subsegments we mean that $H(W_{d_i,\mathcal{A}}\vert C^p_{d_i,\mathcal{A}},C^q_{d_i,\mathcal{A}}) < H(W_{d_i,\mathcal{A}}\vert C^p_{d_i,\mathcal{A}})$ for all $p,q \in [P], p\neq q$. $G_p$ has the same structure as $G$, with the subsegments requested by users not served by server $S_p$ removed. Let the set of nodes in the index coding problem represented by graph $G_p$ be denoted by $\mathcal{I}_p$.

We have the following multi-server index coding bound applying the result in \cite{optimality_2} separately on each of the $P$ index coding problems.
\begin{theorem} \label{multiple server bound}
If the rate tuple $\{r_{1,\mathcal{A}}^1,\ldots, r_{K,\mathcal{A}}^1,\ldots, r_{1,\mathcal{A}}^p,\ldots, r_{K,\mathcal{A}}^p,\ldots, r_{1,\mathcal{A}}^{P},\ldots,r_{K,\mathcal{A}}^{P}\}_{\mathcal{A}\subseteq [K]}$ is achievable for the multi-server index coding problem represented by the set of directed graphs $G_p,p=1,\ldots,P$, under the constraint $\sum_{j=1}^{\rho} n_{\pi(j)} r_{d_i,\mathcal{A}}^{\pi(j)} \geq \vert W_{d_i,\mathcal{A}} \vert$, and inter-file coding is not allowed, then $r_{j,\mathcal{A}}^p=0$ if server $S_p$ does not serve user $U_j$, and
\begin{align}
\sum_{p=1}^P \sum_{\mathcal{J}_p} r_{j,\mathcal{A}}^p \leq 1
\end{align}
for all $\mathcal{J}_p\subseteq \mathcal{I}_p $ where the subgraph of $G_p$ over $\mathcal{J}_p$ does not contain a directed cycle.
\end{theorem}
\begin{remark}
Theorem \ref{multiple server bound} holds when the nodes in the $P$ index coding problems correspond to distinct coded subsegments, that is, there are no repeating nodes in any two index coding problems. Distributed storage schemes which concatenate repetition codes with other storage codes may not satisfy the bound in Theorem \ref{multiple server bound} (for example, see \cite{FR_codes}). References \cite{distr_index_coding} and \cite{cooperative_index_coding} may indicate how to compute the capacity under such distributed storage schemes, but they are outside the scope of this paper. 
\end{remark}
\begin{remark}
There are non-MDS distributed storage codes, called regenerating codes, that utilize increased storage capacity on the servers to reduce the repair bandwidth \cite{dimakis_networkcodes}. Theorem \ref{multiple server bound} holds for them unless some repetition code is used, because the problem can still be represented as $P$ independent index coding problems. For example, Theorem \ref{multiple server bound} holds if a product matrix code \cite{product_matrix} is used for distributed storage. However, a sub-optimal delivery latency is achieved, because each server stores a larger number of packets that have to be transmitted to the connected users for successful file reconstruction.
\end{remark}

To identify the acyclic sets $\mathcal{J}_p$ in the subgraph $G_p$, consider the permutations $\mathbf{u}=(u_1, \ldots, u_K)$ of $[K]$. To determine the tightest bound, we may only consider the largest such sets without a directed cycle. For a given $\mathbf{u}$, the largest set of nodes not containing a directed cycle is 
\begin{align}
    \left\{ C^p_{d_{u_i},\mathcal{A}_i}  : \mathcal{A}_i \subseteq [1:K]\setminus \{ u_1, \ldots, u_i \},i=1,\ldots, K\right\}.
\end{align} Each permutation $\mathbf{u}$ gives a unique acyclic set of nodes of the graph. The subsegment $C^p_{d_i,\phi}$ is not cached in any user, so there is no outgoing edge from $C^p_{d_i,\phi}$ to any other nodes in any sub-index coding problem. Therefore $C^p_{d_i,\phi}$ is always in the set $\mathcal{J}_p$.

Consider first $M_S=\frac{N}{\rho}$. In that case, $\sum_{j=1}^{\rho} n_{\pi(j)} r_{d_i,\mathcal{A}}^{\pi(j)} = \sum_{j=1}^{\rho} \vert C_{d_i,\mathcal{A}}^{\pi(j)} \vert = \vert W_{d_i,\mathcal{A}} \vert$. Following Theorem \ref{multiple server bound}, in order to recover all the desired segments for each user, the deliver latency, $T_{sd}$ must satisfy
\begin{align} 
 FT_{sd} &\geq \sum_{p=1}^{P}\left(\sum_{ \mathcal{A}\subseteq [1:K]\setminus \{u_1\}} \lvert C^{p}_{d_{u_1},\mathcal{A}} \rvert +  \cdots + \sum_{\mathclap{\mathcal{A}\subseteq [1:K]\setminus (\{u^i\}\cap \mathcal{K}_p)  } } \lvert C^{p}_{d_{u_i},\mathcal{A}} \rvert + \cdots + \sum_{\mathclap{ \mathcal{A}\subseteq [1:K]\setminus (\{u^K\} \cap \mathcal{K}_p) }} \lvert C^{p}_{d_{u_K},\mathcal{A}} \rvert \right) \label{eq:bound} \\
 & \ \ \ \  \text{s.t.}\ \ \ \ \sum_{p\in \mathcal{Z}_j} \lvert C^{p}_{d_{u_j},\mathcal{A}} \rvert =  \lvert W_{d_{u_j},\mathcal{A}} \rvert \quad \quad j\in [K],\label{eq:constraint}
\end{align}
for every permutation $\mathbf{u}$, and for every network topology.

We have $\lvert C^{p}_{d,\mathcal{A}} \rvert = \frac{\lvert W_{d,\mathcal{A}} \rvert}{\rho}$ for $M_S=\frac{N}{\rho}$, due to $(P,\rho)$ MDS coded storage. In Eq. \eqref{eq:bound}, in the summation for a fixed value of $p$, the number of terms with $\vert \mathcal{A} \vert =i$ is $ {K\choose i+1}-{K-q_p\choose i+1}$. Thus we have
\begin{align} 
 FT_{sd} \geq &  \sum_{p=1}^{P} \sum_{i=0}^{K-1} \frac{\left( {K\choose i+1}-{K-q_p\choose i+1}\right)}{{K\choose i}} x^{p}_i \\
= & \sum_{p=1}^{P} \sum_{i=0}^{K-1} \frac{\left( {K\choose i+1}-{K-q_p\choose i+1}\right)}{\rho {K\choose i}} x_i\\
= &  \sum_{i=0}^{K-1} \left[ \sum_{p=1}^{P} \frac{\left( {K\choose i+1}-{K-q_p\choose i+1}\right)}{\rho {K\choose i}}\right] x_i \label{lower_bound}\\
\text{while} \ \ & x_0 + x_1 + \cdots + x_K \geq F , \label{constraint1}\\
\text{and} \ \  & x_1 + 2x_2 + \cdots + Kx_K \leq \frac{KM_U}{N}F \label{constraint2}
\end{align}
where $x_i\triangleq \sum_{\mathcal{A} \subset [K]: \lvert \mathcal{A} \rvert = i} \lvert W_{j,\mathcal{A}} \rvert  = {K \choose i} \lvert W_{j,\mathcal{A}} \rvert = \rho {K \choose i} \lvert C^p_{j,\mathcal{A}} \rvert$ is the total normalized size of all segments of file $j$ cached by $i$ users; or equivalently, $x_i^p \triangleq  {K \choose i} \vert C_{j,\mathcal{A}}^p \vert = \frac{1}{\rho} x_i $ is the total normalized size of all subsegments of file $j$ cached by $i$ users and stored in server $S_p$.
We minimize the lower bound in Eq. \eqref{lower_bound} over all segment sizes $x_i$, which is a linear program with two linear constraints \eqref{constraint1} and \eqref{constraint2}, where the former follows from the sum of all fractions of the files being one, while the latter follows from the user cache memory constraint. The solution of a linear program lies on one of the corner points of the feasible region. The feasible region defined by the constraints has only one corner point characterized by 
\begin{align*}
    x_i=\left\{  \begin{array}{cc}
        F & i=t, t=\frac{KM_U}{N} \\
        0 & \text{otherwise}
    \end{array}
    \right. .
\end{align*}
Therefore, Eq. \eqref{lower_bound} simplifies to
\begin{align}
T_{sd} \geq \sum_{p=1}^{P} \frac{\left( {K\choose t+1}-{K-q_p\choose t+1}\right)}{\rho{K\choose t}},
\end{align}
which is achieved by our delivery scheme. This proves the optimality of the delivery scheme for successive transmission proposed in Section \ref{s:Successive} under the assumption of MDS coded storage at the servers and uncoded caching at the users.

\subsection{Redundancy in server storage}
When there is redundant server storage capacity, i.e., server storage capacity is 
$M_S=\frac{N}{\rho-z}$, consider the constraint $\sum_{j=1}^{\rho} n_{\pi(j)} r_{d_i,\mathcal{A}}^{\pi(j)} \geq \vert W_{d_i,\mathcal{A}} \vert$ in Theorem \ref{multiple server bound}. Since the bound in Theorem \ref{multiple server bound} is a linear program of the rates of transmission of the coded subsegments from the servers, the optimal solution lies on one of the corner points of the feasible region defined by the constraint. The corner points for $M_S=\frac{N}{\rho-z}, z\in [\rho-2]$, are those where $ n_{\pi(j)} r_{d_i,\mathcal{A}}^{\pi(j)}=\frac{ \vert W_{d_i,\mathcal{A}} \vert}{\rho-z}$ for all $j\in \mathcal{R}, \mathcal{R}\subset \mathcal{Z}_i, \vert \mathcal{R} \vert = \rho - z$, and equal to $0$ for all $j \in \mathcal{Z}_i \setminus \mathcal{R}$. The optimal solution should lie on the corner point which chooses $\mathcal{R}$ such that the servers in $\mathcal{R}$ have the most multicasting opportunities, and can thus deliver $\rho-z$ coded subsegments of the requested segments to the users in the minimum number of transmissions. This is equivalent to finding the minimum cover for each multicast group as described in Section \ref{sec:incservercache}.

When fractional repetition (FR) codes are used for server storage \cite{fr_codes}, the minimum cover scheme may not be optimal. However, since FR codes have a maximum code rate of $\frac{1}{2}$, we cannot have distributed storage schemes where $M_S \leq \frac{2N}{P}$. Thus the minimum cover scheme is optimal for server storage capacities $M_S \leq \frac{2N}{P}$. We illustrate with a toy example that the bound in Theorem \ref{multiple server bound} does not hold when FR codes are used.
\begin{example}\label{toy_ex}
Consider the simple scenario with $P=2$ servers, $K=3$ users illustrated in Fig. \ref{fig:toy}, where we assume each server can store all the $N=3$ files, i.e., $M_S=3$, and each user has cache capacity $M_U=1$. Let the cache contents of $U_1, U_2, U_3$ be $W_2, W_3, W_1$, respectively, and the demand vector $\mathbf{d}=\{ W_1, W_2, W_3\}$. In this example, the demands can be satisfied by $S_1$ transmitting $W_1 \oplus W_2$, and $S_2$ transmitting $W_1 \oplus W_3$, that is, the delivery latency of $T_{sd}=2$ is achievable. However, Theorem \ref{multiple server bound} gives the bound on delivery latency as $T_{sd}\geq 3$. $U_2$ receives its requested file $W_2$ with added interference of $W_1$ from $S_1$, which it cannot remove using its cache contents. However, $U_2$ adds the messages from $S_1$ and $S_2$ to align the interference on $W_2$ with $W_3$, which it can remove by using its cache contents, thus doing a sort of interference alignment. In contrast, if MDS coded storage were used, the interference alignment type of scheme would not be possible due to distinct coded subsegments transmitted by both servers.
\begin{figure}[htbp]
    \centering
    \includegraphics[scale=0.6]{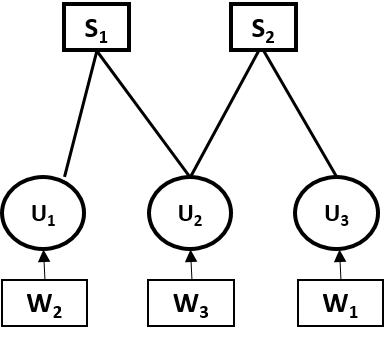}
     \caption{Toy example with $P=2, K=N=M_S=3, M_U=1$}
    \label{fig:toy}
\end{figure}
\end{example}
The polymatroidal capacity region for multi-server index coding has been characterized in \cite{distr_index_coding} for full user-server connectivity and uncoded server storage. Characterizing the capacity region for partial user-server connectivity, and constructing an optimal joint server storage and caching scheme for FR coded distributed storage is an interesting open problem for future work.

\section{Parallel SBS Transmissions}\label{s:parallel}
When SBSs can deliver in parallel without interfering with each other, the normalized delivery latency is dictated by the SBS that has to deliver the maximum number of bits:
\begin{align}\label{eq:link_load}
T_{pd} \triangleq \max_{q_p} \frac{1}{\rho {K\choose t}} \left[{K\choose t+1}-{K-q_p\choose t+1} \right].
\end{align}

The ``best'' and ``worst'' network topologies in the parallel transmission scenario are different from those in the successive transmission scenario. The most balanced topology, i.e., the one with the minimum value of the maximum $q_p$ has the ``best'' (lowest) delivery latency, contrary to the successive transmission scenario, in which this would be the ``worst'' topology. The corresponding delivery latency can be obtained by substituting $q_p=\lceil \frac{K\rho}{P} \rceil$ in \eqref{eq:link_load}.  The topology with the maximum possible $q_p$, i.e., any topology with at least one server connected to all $K$ users, is the ``worst'' topology since it has the highest delivery latency.  
\subsection{Redundant server storage capacity}
The minimum server storage capacity that would allow the reconstruction of any demand combination is given by $M_S=\frac{N-M_U}{\rho}$. In this case, we cache the same $\frac{M_U}{N}$ fraction of the library in all the user caches during the placement phase, and deliver the remaining fraction of the demands from the servers without multicasting. The worst-case delivery latency in this case is $T_{pd}=\frac{K}{\rho}\left(1-\frac{M_U}{N} \right)$.

Next, we consider the case when there is redundancy in server memories; that is, $\frac{N}{\rho}< M_S \leq N$. Assume that $M_S=\frac{N}{\rho - z}$ for some integer $ z\in [\rho -1]$. For non-integer values of $z$, the solution can be obtained by memory-sharing. Notice that, as for successive transmissions, users can select the servers from which to receive coded subsegments. A greedy server allocation algorithm is used. The algorithm assigns the multicast messages to the servers trying to keep the number of messages delivered by each server as evenly distributed as possible. At any point in time, if a server has delivered a higher number of messages than all the other servers, even if a better multicasting opportunity is available to this server, that server is not assigned a multicast message in order to balance the number of messages delivered by each server in a greedy manner. Instead, the server with the next best multicasting opportunity and a smaller count of transmissions is assigned to transmit a particular coded subsegment to a multicast group. Compare this with the algorithm for successive transmission, where a multicast message is always assigned to the server with the maximum multicasting opportunity. It is easy to see that the delivery latency achieved depends on the order in which the algorithm assigns multicast messages to the servers. Thus the proposed algorithm is suboptimal. Numerical results illustrating the performance of the proposed delivery algorithm will be presented in the next section.

\section{Results and Discussions} \label{s:Discussion}

\begin{figure}[h]
    \centering
    \includegraphics[width=5.2cm, height=4.7cm]{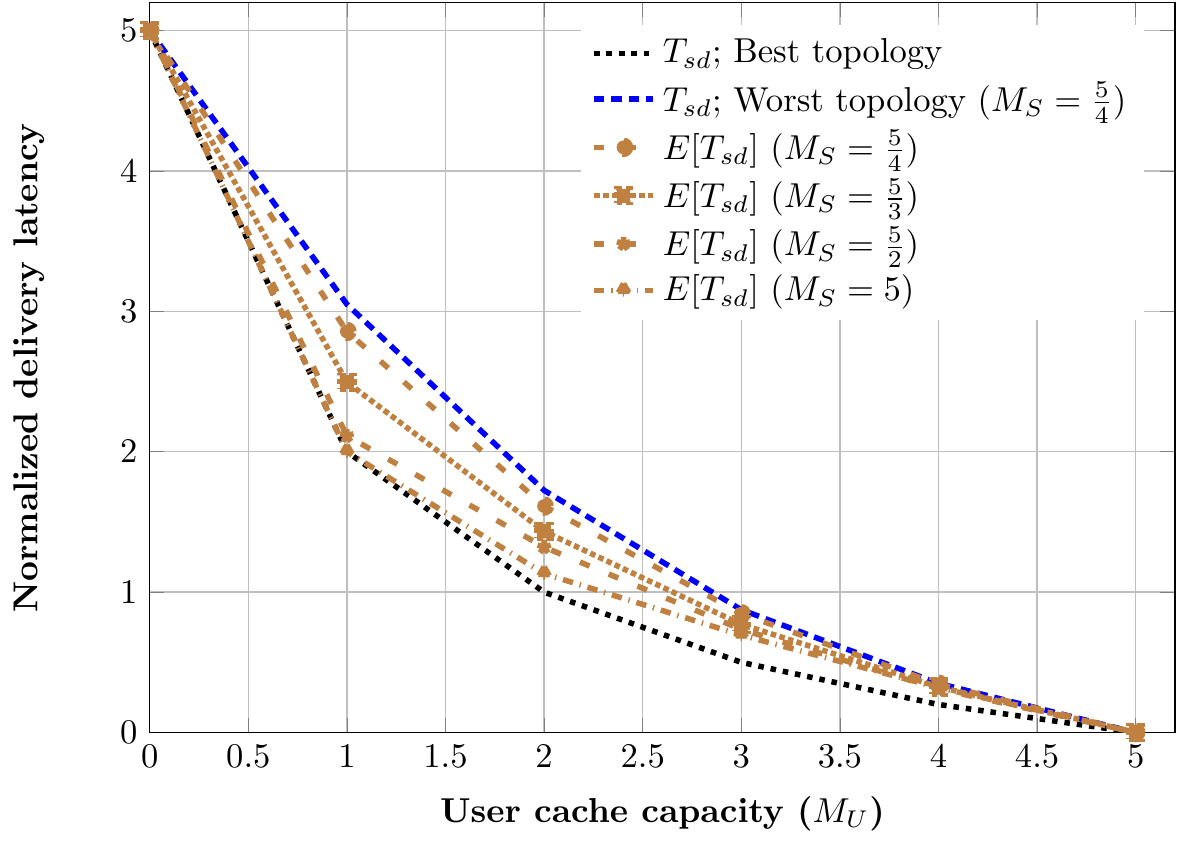}
\caption{\small Average normalized delivery latency vs. user cache capacity $M_U$, for $P=7, N=K=5, \rho=4$, and for server storage capacities of $M_S=\frac{5}{4}, \frac{5}{3}, \frac{5}{2}, 5$. }\label{fig:f3}
\end{figure}
In Fig. \ref{fig:f3} we plot the achievable trade-off between the user cache capacity and the normalized delivery latency, $T_{sd}$, for the best and worst topologies, and the average normalized delivery latency over all topologies, for successive transmission. 
The trade-off curves are plotted for different server storage capacities. We observe that the gap between the worst and the best topologies can be significant. From \eqref{eq:sumrate} and \eqref{eq:avgrate} we can deduce that, for successive transmission the worst topology delivery latency; and hence, the average delivery latency of the proposed scheme are both within a multiplicative factor of $\frac{1}{\alpha}$ of the best topology delivery latency. We observe from Fig. \ref{fig:f3} that the delivery latency decreases significantly, particularly for low $M_U$ values, as the redundancy in server storage increases. 
\begin{figure}
    \centering
\includegraphics[width=5.2cm, height=4.7cm]{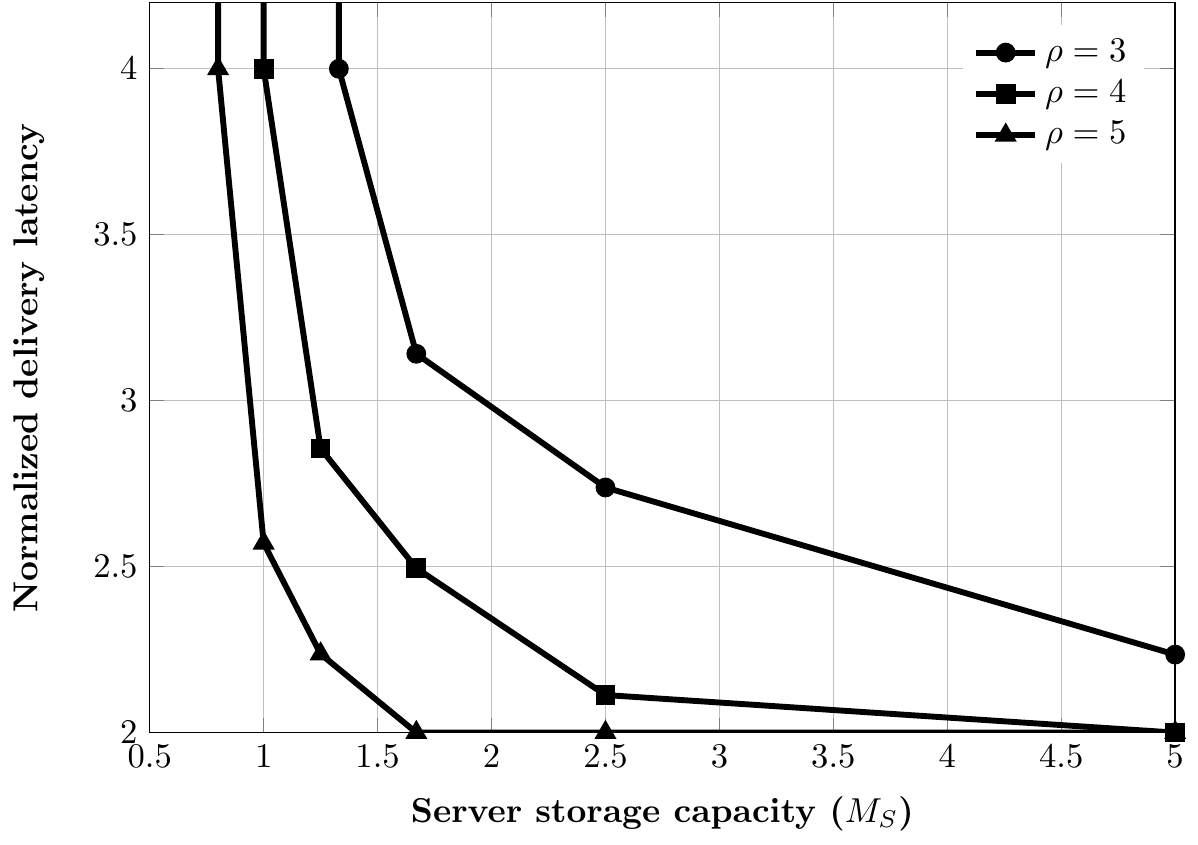}
\caption{\small Average normalized delivery latency vs. server storage capacity $M_S$, for $P=7, N=K=5, M_U=1$ for successive SBS transmissions. }\label{fig:f4}
\end{figure}

In Fig. \ref{fig:f4} the average delivery latency for successive transmission is plotted as a function of the server storage capacity for server storage capacities $M_S\in [\frac{N-M_U}{\rho}, N]$. The figure is obtained by performing Monte Carlo simulations with uniform random realizations of the topology and averaging the delivery latency over them. We observe from Fig. \ref{fig:f4} that the average delivery latency decreases rapidly for an initial increase in the server storage capacity, which is more significant for high $\rho$ values. This is because, thanks to MDS-coded storage at the servers, the number of available multicasting opportunities increases with the redundancy across servers. Fig. \ref{fig:f4} highlights the fact that, for successive delivery and sufficient network connectivity, increasing the server storage beyond a certain value has little or no impact on the delivery latency.

In Fig. \ref{fig:f7}, it is shown that Eq. \eqref{closedform_redundant} in Section \ref{performance_analysis} gives a fairly accurate estimate of the expected delivery latency for successive transmissions with redundant server storage capacity, especially for small server storage capacities. The theoretical estimate diverges a little from the expected rate for large server storage capacity, before again converging where the delivery latency saturates at the minimum. Also, comparing the plot for $\rho=9, P=21$ in Fig. \ref{fig:f7} with the plot for $\rho=3, P=7$ in Fig. \ref{fig:f4}, where the connectivity $\alpha$ is the same, we observe that the average delivery latency decreases faster for $\rho=9, P=21$; that is, for larger values of $P$.

\begin{figure}[h]
    \centering
\includegraphics[width=5.5cm, height=5.0cm]{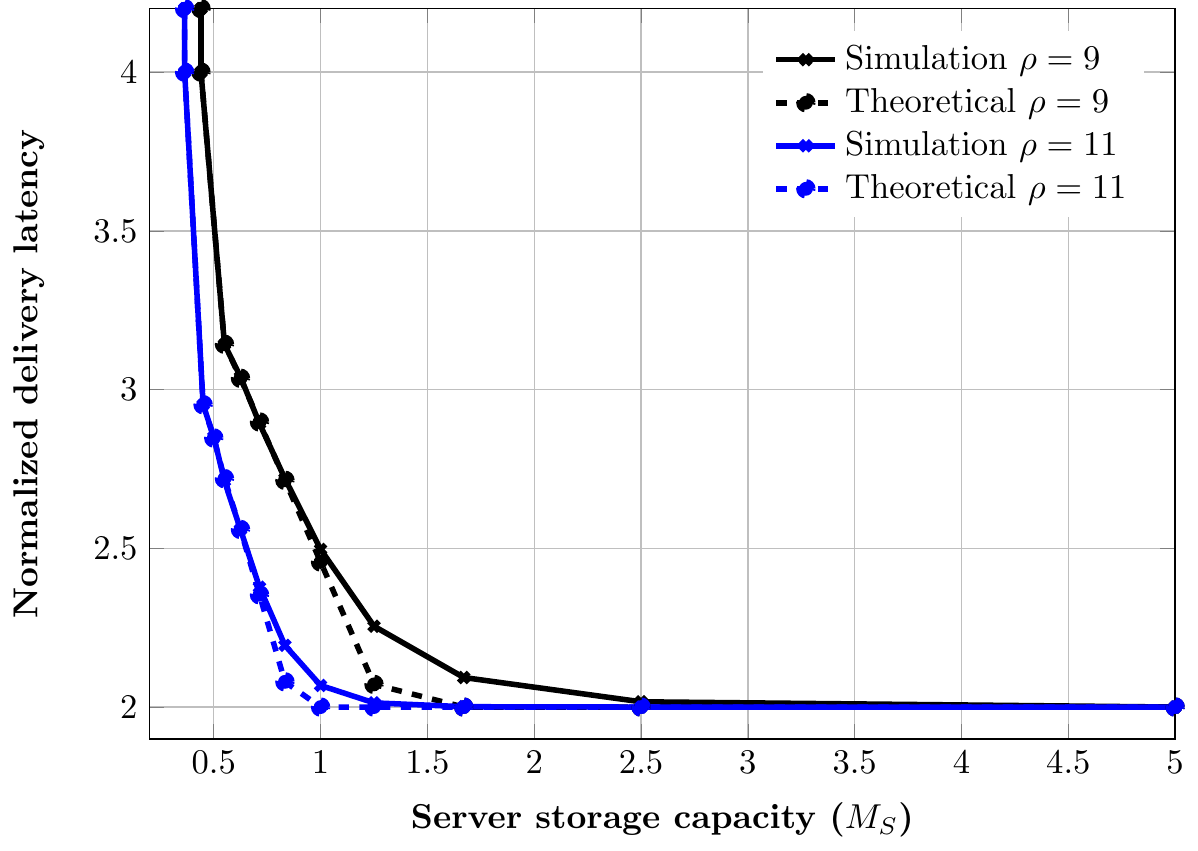}
\caption{\small Comparing the simulation with the theoretical, Average normalized delivery latency vs. server storage capacity, for $P=21, N=K=5, M_U=1$ for successive SBS transmissions. }\label{fig:f7}
\end{figure}

\begin{figure}
    \centering
\includegraphics[width=4.7cm, height=4.2cm]{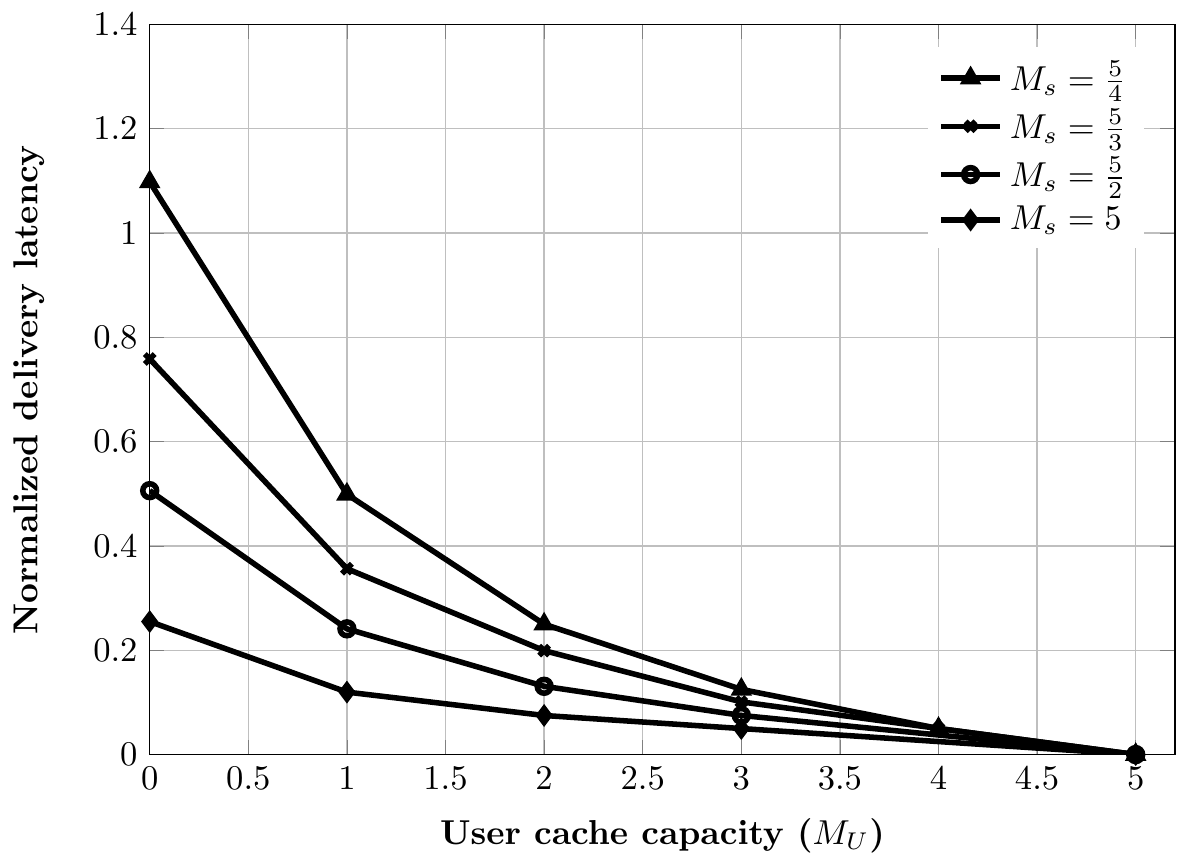}
\caption{\small Average normalized delivery latency for parallel transmissions vs. user cache capacity $M_U$, for $P=7, N=K=5, \rho=4$, and for server storage capacities of $M_S=\frac{5}{4}, \frac{5}{3}, \frac{5}{2}, 5$. }\label{fig:f6}
\includegraphics[width=4.7cm, height=4.2cm]{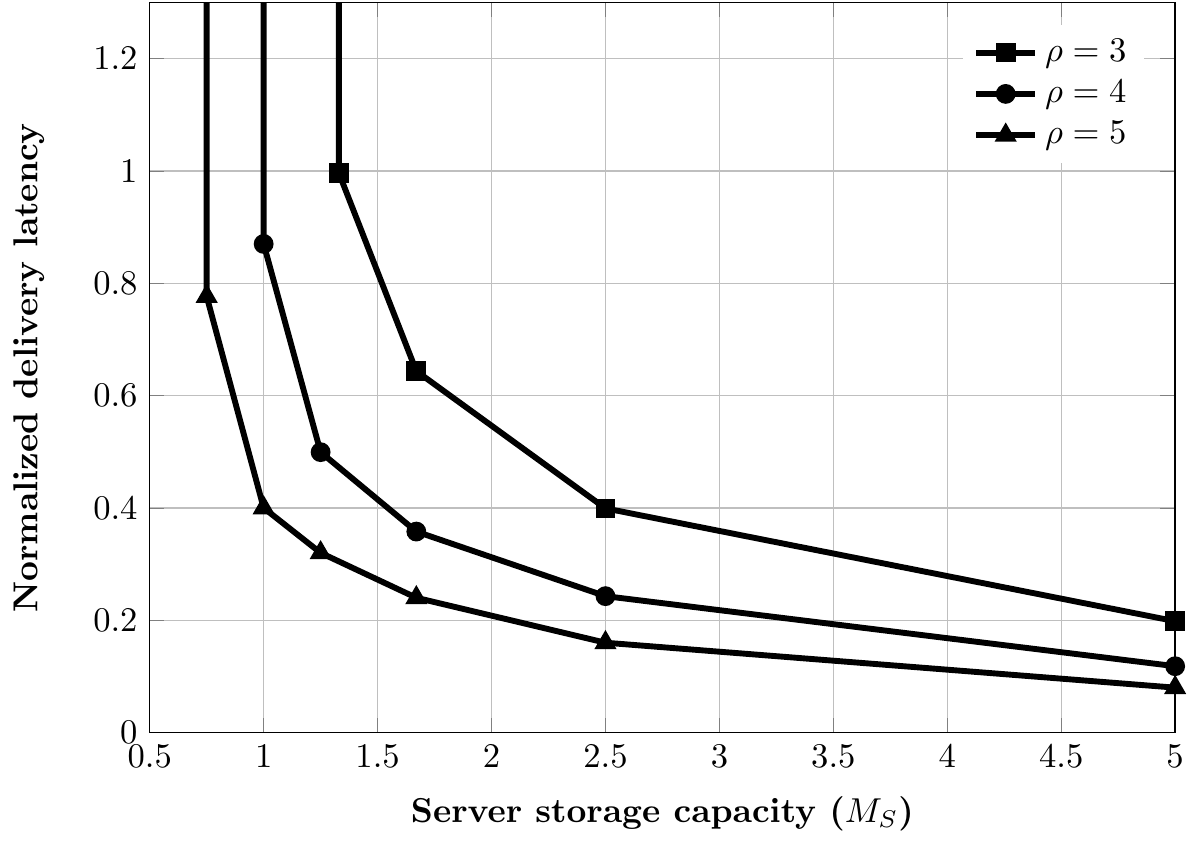}
\caption{\small Average normalized delivery latency vs. server storage capacity $M_S$, for $P=7, N=K=5, M_U=1$ for parallel transmissions. }\label{fig:f5}
\end{figure}

The average delivery latency for parallel transmissions is plotted with respect to the user cache capacity in Fig. \ref{fig:f6}, using \eqref{eq:link_load}. We observe as before that increasing the server storage capacity gives significant gains in the average delivery latency, especially for low values of $M_U$. Unlike the case for successive transmissions, the average delivery latency for $M_U=0$ also reduces as the server storage capacity is increased.

The average delivery latency for parallel transmissions is plotted with respect to the server storage capacity, $M_S$, in Fig. \ref{fig:f5}. Unlike the delivery latency for successive transmissions, we can see that the delivery latency does not saturate, and keeps decreasing until all the files are stored at each of the servers. We also observe as before that the increase in network connectivity $\alpha$ helps reduce the delivery latency significantly, especially for low server storage capacity $M_S$.

\section{Conclusions and future work}

We have studied a multi-server coded caching and delivery network, in which cache-equipped users connect randomly to a subset of the available servers, each with its own limited storage capacity. While this allows each server to have only a limited amount of storage capacity, it requires coded storage across servers to account for the random topology. We proposed a joint coded storage, caching and delivery scheme that jointly applies MDS-coded storage at the servers, and uncoded caching and coded delivery to the users. The achievable delivery latency of this scheme for both successive and parallel transmissions from the SBSs are presented, with increasing user cache memory as well as increasing server storage capacity, and their averages over random network topologies are plotted. The analysis shows that when the server storage capacity is increased, the delivery latency can be reduced significantly, for both successive transmissions as well as parallel transmissions. However, it is also observed that for sufficient network connectivity, increasing server storage beyond a certain value provides little benefit. Increasing server storage has a more significant impact when there is low connectivity, and when user cache capacities are small.

An interesting open problem for future work is finding a lower bound and an optimal scheme when FR codes are used for distributed storage. A toy example \ref{toy_ex} is given in this paper which illustrates the potential benefits of such codes. The toy example also presents an asymmetry in the user connections to the servers, where users 1 and 3 connect to one server each, while user 2 connects to 2 servers. An interesting problem is constructing a general scheme for heterogeneous network topologies and extracting gains from such topologies as demonstrated in the toy example. Another question relates to gains from heterogeneous distributed storage. For instance, if there is knowledge of user dynamics and non-uniform probabilities of the user-server connections, can a heterogeneous distributed storage scheme be designed to extract higher average gains? An extreme case of this scenario would mimic the combination network model where the user-server connections are completely fixed and known, which achieves higher gains. Such open problems present ripe material for future research.


\bibliographystyle{IEEEtran}


\end{document}